\documentclass[conference]{IEEEtran}
\IEEEoverridecommandlockouts

\usepackage[utf8]{inputenc} 
\usepackage[T1]{fontenc}
\usepackage{url}
\usepackage{ifthen}
\usepackage{cite}
\usepackage[cmex10]{amsmath}
\usepackage{amssymb,amsfonts}
\usepackage{amsthm}
\usepackage{algorithmic}
\usepackage{graphicx}
\usepackage{textcomp}
\usepackage{xcolor}
\usepackage{pgfplots}
\usepackage{tikz}
\def\BibTeX{{\rm B\kern-.05em{\sc i\kern-.025em b}\kern-.08em
    T\kern-.1667em\lower.7ex\hbox{E}\kern-.125emX}}

\newcommand{\F}{\mathbb{F}}
\newcommand{\E}{\mathbb{E}}
\newcommand{\mC}{\mathcal{C}}

\usepackage{amsthm}
\theoremstyle{definition}
\newtheorem{theorem}{Theorem}

\newtheorem{lemma}{Lemma}
\newtheorem{remark}{Remark}
\newtheorem{definition}{Definition}

\newtheorem{proposition}{Proposition}
\newtheorem{notation}{Notation}
\newtheorem{example}{Example}
\newtheorem{problem}{Problem}

\begin{document}

\title{The Coverage Depth Problem in DNA Storage \\  Over Small Alphabets 
\thanks{The research of M.B. is partially supported by the EuroTech Program. The research of A.R. is supported by the Dutch Research Council through grants VI.Vidi.203.045, 
OCENW.KLEIN.539, and by the European Commission. The research of E.Y. is funded by the European Union (ERC, DNAStorage, 101045114). Views and opinions expressed are however those of the authors only and do not necessarily reflect those of the European Union or the European Research Council Executive Agency. Neither the European Union nor the granting authority can be held responsible for them.} 
}

\author{%
   \IEEEauthorblockN{\textbf{Matteo Bertuzzo}\IEEEauthorrefmark{1}, \textbf{Alberto Ravagnani}\IEEEauthorrefmark{1}, and \textbf{Eitan~Yaakobi}\IEEEauthorrefmark{2}}
    \IEEEauthorblockA{\IEEEauthorrefmark{1}%
    Dept. of Mathematics and Computer Science, Eindhoven University of Technology, the Netherlands}
    \IEEEauthorblockA{\IEEEauthorrefmark{2}%
    Faculty of Computer Science, 
    Technion--Israel Institute of Technology, Israel}
    \texttt{\{m.bertuzzo, a.ravagnani\}@tue.nl, yaakobi@cs.technion.ac.il   \vspace{-1ex}}
 \vspace{-3ex}}

\maketitle

\begin{abstract}
The coverage depth problem in DNA data storage is about minimizing the expected number of reads until all data is recovered. When they exist,
MDS codes offer the best performance in this context. 
This paper focuses 
on the scenario where the base field is not large enough to allow the existence of MDS codes. We investigate the performance 
for the coverage depth problem
of codes defined over a small finite field, providing closed formulas for the expected number of reads for various code families. We also compare the results with the theoretical bounds in asymptotic regimes.
The techniques we apply range from probability, to duality theory and combinatorics.
\end{abstract}

\section{Introduction}
The volume of data generated globally is growing at an exponential rate, creating an escalating demand for storage that outpaces the current supply~\cite{b1}. This motivates the urgent need for innovative, high-density, efficient, and durable storage solutions that outperform existing technologies. In this context, DNA-based storage systems emerge as a promising alternative to traditional storage media, particularly for long-term data archiving, due to their exceptional density, durability, and low maintenance costs~\cite{b2},~\cite{b3}.

To store data in DNA, a multi-step process is carried out: first, the original data is encoded, transforming it from a string of bits into sequences based on the DNA alphabet ${A, C, G, T}$. These sequences are then divided into blocks, and in the synthesis step, artificial DNA molecules, or \emph{strands}, are produced, with multiple copies of each strand being generated. Once synthesized, the DNA strands are stored in a container.

When a user wants to retrieve the stored information, the sequencing process is performed: the strands are translated back into DNA sequences, called \emph{reads}, which are copies of the previously synthesized strands and may contain errors. Finally, these sequences are decoded to recover the user's original information.
One key distinction between retrieving data from DNA and traditional storage media lies in this step: the DNA strands are read \emph{randomly}.

Although DNA has high potential as a storage medium~\cite{b4, b5, b6, b7, b8, b9, b10}, the slow throughput and high costs compared to alternative storage techniques, resulting from the efficiency of DNA sequencers, represent a drawback~\cite{b2, b11, b12}. This problem is related to the concept of the \emph{coverage depth}~\cite{b13}, defined as the ratio between the number of sequenced reads and the number of designed DNA strands.

This paper focuses on the \emph{coverage depth problem}, recently introduced in~\cite{cover_your_bases}, which involves minimizing the number of reads required to retrieve a particular piece of data encoded in DNA. We concentrate on the case where all the user's original information should be recovered entirely. Suppose we have $n$ encoded strands starting from $k$ information strands representing the original data: taking inspiration from the coupon collector's problem~\cite{ccp1, ccp2, ccp3, ccp4}, if the $k$ information strands are encoded by an MDS code, the expected number of reads to decode all the information strands is $n(H_n - H_{n-k})$, where $H_i$ is the $i$-th harmonic number. This result is the best we can achieve in minimizing the expected number of reads.

The scenario of the DNA coverage depth problem where the entire information must be recovered was extended in~\cite{ext1, ext2} to support the setup of the combinatorial composite of DNA shortmers~\cite{ext3}, and in~\cite{ext4} for the setup of composite DNA letters~\cite{b9}.

Motivated by the results and observations in~\cite{cover_your_bases}, in this work we focus on finding closed formulas for the expected number of reads needed to recover the entire information and, consequently, on finding codes that are optimal for minimizing the expected number of reads. The remainder of this paper is organized as follows. In Section~\ref{problem_statement}, we formally define the two problems that we will discuss throughout the paper. Section~\ref{simplex_section} provides a closed formula for computing the expectation for simplex codes. In Section~\ref{duality_section}, we show an important and general duality result, which is then used in Section~\ref{Hamming_section} to obtain a closed formula also for the expectation for Hamming codes. In Section~\ref{estimates_section}, we study the asymptotic behavior of the formulas obtained previously.

\section{Problem Statement} \label{problem_statement}
In this paper, $q$ is a prime power and $\mathbb{F}_q$ is the finite field with $q$ elements. We let $k$ and $n$ be positive integers with $2 \leq k \leq n$. Furthermore, for a positive integer $m$, we 
denote by $H_m$ the $m$-th harmonic number: 
$$H_m=\sum_{i=1}^m \frac{1}{i}.$$

In a typical DNA-based storage system, the data is stored as a length-$k$ vector whose entries are strands of length $\ell$ over the alphabet $\Sigma = \{ A,C,G,T \}$. In particular, the encoded strands are elements of
$(\Sigma^\ell)^k$. To allow using coding theory tools, we embed $\Sigma^\ell$ into $\mathbb{F}_q$ and use a $k$-dimensional linear block code $\mathcal{C} \subseteq \mathbb{F}_q^n$ to encode an information vector $(x_1, \dots, x_k) \in \mathbb{F}_q^k$ to an encoded vector $(y_1, \dots, y_n) \in \mathbb{F}_q^n$. 

If a user wishes to retrieve the stored information, the strands initially undergo an amplification process, followed by sequencing: all this generates various copies for each string, which may contain errors compared to the originals. These are called ``reads''. To simplify our analysis, in this paper we will assume that no errors are made in any of these steps, hence the final output of the process is a multiset of reads, obtained without a specific order. One of the most crucial goals to achieve is the reduction of coverage depth with regards to information retrieval: this would allow to increase the efficiency of DNA sequencing, making it perform better than other storage solutions.

The starting point of this paper is a result about the coverage depth problem for DNA data storage~\cite{cover_your_bases}, when all information strands need to be recovered. Since the $k$ information strands are encoded using a generator matrix~$G \in \mathbb{F}_q^{k \times n}$, there is a one-to-one correspondence between the encoded strands and the columns of~$G$, namely the $i$-th encoded strand corresponds to the $i$-th column of the generator matrix; therefore, recovering the $i$-th information strand is equivalent to recovering the $i$-th standard basis vector, that is, it must be in the span of the already recovered columns of $G$, since we can see these columns as vectors in $\mathbb{F}_q^k$. Motivated by these results, we  define the first problem studied in this paper.

\begin{problem}[\textbf{The Coverage Depth Problem}] \label{first_problem}
    Let $G \in \mathbb{F}_q^{k \times n}$ have rank~$k$. Suppose that the columns of $G$ are drawn uniformly randomly with repetition, meaning that each column can be drawn multiple times. Compute the expected number of columns one needs to draw until all the standard basis vectors are in their $\mathbb{F}_q$-span (or equivalently until the drawn columns have rank~$k$). We denote such expectation by $\mathbb{E}[G]$.
\end{problem}

We start by showing that, in contrast to the random access coverage depth problem~\cite[Problem~1]{a_combinatorial_perspective}, 
where only a single information strand is to be recovered, the value of $\E[G]$ only depends on the row-space of $G$, i.e., on the code that the matrix~$G$ generates.

\begin{proposition} \label{independence_from_G_proposition}
    Let $G,G' \in \F_q^{k \times n}$ have the same row-space. Then $\E[G]=\E[G']$.
\end{proposition}
\begin{proof}
Since $G$ and $G'$ have the same row-space, there exists an invertible matrix $A \in \F_q^{k \times k}$ with $G'=AG$. 
The statement follows from the fact that multiplying by $A$ preserves the linear dependencies among columns.
\end{proof}

The previous result shows that the quantity 
$\E[\mC]$ is well defined for a linear error-correcting code $\mC \subseteq \F_q^n$ as $\E[G]$, where $G$ is \emph{any} generator matrix of $\mC$. Therefore, we will use the symbols $\E[G]$ and $\E[\mC]$ interchangeably. We are now ready 
to state
the second problem studied in this paper.

\begin{problem}[\textbf{The Optimal Coverage Depth Problem}] \label{second_problem}
    For given values of $n$, $k$ and $q$, compute the value
$$\mathbb{E}_{\textnormal{opt}}[n,k]_q \triangleq \min\{\mathbb{E}[\mC] \, : \,  \mC \text{ is an $[n,k]_q$ code}\},$$
    and construct
    a code $\mC$ attaining the minimum. 
\end{problem}

Throughout the paper, $G \in \mathbb{F}_q^{k \times n}$  denotes a rank $k$ matrix and $\mC \subseteq \mathbb{F}_q^n$ is the $k$-dimensional code generated by $G$.
Several results on Problems~\ref{first_problem} and~\ref{second_problem} were obtained in~\cite{cover_your_bases}. We mention the most important ones for our purposes.

\begin{theorem}[\text{see \cite[Corollary~2]{cover_your_bases}}] \label{thm_MDS_bound}
    For any generator matrix~$G$ of an $[n,k]_q$ code $\mC$ we have 
    \[
    \E[G] \ge \sum_{i=0}^{k-1} \frac{n}{n-i} = n(H_n - H_{n-k}).
    \]
    Furthermore, the lower bound is attained with equality only by any generator matrix of an MDS code. 
\end{theorem}

\begin{theorem}[\text{see \cite[Theorem~2]{cover_your_bases}}] \label{thm_asintotic_MDS}
    Let $R$ be a constant, $0<R<1$, and for all $n$ let $\mC_n$ be an $[n, k_n=\lfloor nR \rfloor ]_{q_n}$ MDS code. We have 
    \[
    \lim_{n \rightarrow\infty} \frac{\mathbb{E}[\mC_n]}{k_n} = \frac{1}{R} \log\Big( \frac{1}{1-R} \Big).
    \]
    Furthermore, consider a sequence of MDS codes $\{ \mC_i \}_{i=1}^\infty$ with parameters $n_i, k_i$ such that $\lim_{i \rightarrow\infty} k_i/n_i = 0$. Then
    \[
    \lim_{i\rightarrow\infty} \frac{\mathbb{E}[\mC_i]}{k_i} = 1.
    \]
\end{theorem}

\begin{remark}
    Theorem~\ref{thm_MDS_bound} provides a lower bound on the expectation and solves Problem~\ref{first_problem} for MDS codes. It also solves Problem~\ref{second_problem} for any choice of parameters $n, k$ and $q$ such that there exists an $[n,k]_q$ MDS code. In particular, assuming that the MDS conjecture~\cite{mds1} holds, we can write
    \[
    \mathbb{E}_{\text{opt}}[n,k]_q = n(H_n - H_{n-k}) \text{ when } q \ge n-1.
    \]
    Lastly, Theorem~\ref{thm_asintotic_MDS} gives the asymptotic value for the minimum expectation.
\end{remark}

It is well known that MDS codes only exist over sufficiently large finite fields. In fact, it has been conjectured (and proven in several instances) that $q \ge n-1$ is a necessary condition for the existence of an MDS code, with the exception of very few parameter sets that require ``only'' $q \ge n-2$; see~\cite{mds1,mds2,mds3,mds4}.
Given the fact that field size and code length are imposed by the storage scheme setup, and it may not be possible to choose them, it is therefore natural to investigate what results can be achieved in this context when $q$ is too small to allow the existence of an MDS code.
That's the focus of this paper.

\section{Performance of the Simplex Code} \label{simplex_section}
When focusing on small finite fields, it is natural to consider simplex codes. 
The simple structure of any generator matrix of a simplex code makes it possible to obtain a closed formula for $\E[G]$ using the $q$-analogue of a standard argument for the coupon collector's problem, solving Problem~\ref{first_problem} for this family of codes.

\begin{theorem} \label{thm_E_k_simplex}
    Let $\mathcal{C} \subseteq \mathbb{F}_q^n$ be the $q$-ary simplex code of dimension $k$, where $n = (q^k-1)/(q-1)$. We have 
    \begin{equation*} 
    \mathbb{E}[\mathcal{C}] = k + \sum_{i=1}^k \frac{q^{i-1}-1}{q^k - q^{i-1}}.
    \end{equation*}
\end{theorem}
\begin{proof}
    Fix any generator matrix $G$ of $\mC$. For $i \in \{1, \ldots, k\}$, let $s_i(\mC)$ be the random variable that governs the number of draws until the selected columns span a space of dimension $i$, when the columns previously drawn span a space of dimension $i-1$. Note that the expected value of $s_1(\mC)$ is equal to 1, since all columns of $G$ are non-zero.

    Since the columns of $G$ are the elements of $\F_q^k$ up to multiples, 
    $s_i(\mathcal{C})$ is a geometric random variable with success probability
    \[
    p_i = \frac{n - \frac{q^{i-1}-1}{q-1}}{n}.
    \]
    By the linearity of expectation and $\E[s_i(\mathcal{C})] = 1/p_i$, we therefore have
    \[
    \begin{split}
    \mathbb{E}[\mathcal{C}] &= \mathbb{E}\left[ \sum_{i=1}^k s_i(\mathcal{C}) \right] = \sum_{i=1}^k \mathbb{E}[s_i(\mathcal{C})] 
    = k + \sum_{i=1}^k \frac{q^{i-1}-1}{q^k - q^{i-1}},
    \end{split}
    \]
    as desired.
\end{proof}

It is natural to compare the values obtained for simplex codes with the lower bound of Theorem~\ref{thm_MDS_bound}: we will discuss this in Section~\ref{estimates_section}. 

We conjecture that simplex codes minimize the expectation among all codes with $n = (q^k-1)/(q-1)$, for given $q$ and~$k$. While we do not have a general proof for this result, we present some computational evidence.

\begin{example}
We ask ourselves which code $\mC$ minimizes the value $\E[\mC]$, among all $[7,3]_2$ codes, i.e., what is the value of $\mathbb{E}_{\textnormal{opt}}[7,3]_2$.
Note that we may restrict our search 
to codes whose generator matrices do not have any zero column. Indeed, if a code $\mathcal{C}$ has a zero column in one (and thus all) generator matrix,
then we can replace that zero column with any non-zero vector of $\F_2^3$. It is easy to see that the value of $\E[G]$ can only decrease this way. 
We computationally checked all codes whose generator matrices have non-zero vectors as columns, and found that
the best result is indeed obtained by a generator matrix of the simplex code. More precisely, if $\mC$ is not the simplex code, then $\E[\mC] \geq 17/4 = 4.25$, while for the simplex code we have $\mathbb{E}[\mC] = 47/12 \approx 3.917$. Therefore, simplex codes solve Problem~\ref{second_problem} for this given choice of parameters.
\end{example}

\section{A Duality Result} \label{duality_section}
Computing $\E[\mC]$ was relatively easy for a simplex code~$\mC$, thanks to the structure of its generator matrix. However, the computation is significantly more challenging for an arbitrary code. In this section, we establish a general duality result that expresses $\E[\mC]$ in terms of the combinatorial structure of the dual code $\mC^\perp$. In Section~\ref{Hamming_section}, we will apply our result to compute the value $\E[\mC]$ when $\mC$ is a Hamming code.

\begin{definition}
    Let $\mathcal{C} \subseteq \F_q^n$ be a $k$-dimensional code. A non-empty set $S \subseteq \{1, \dots, n \}$ is an \textbf{information set} for $\mC$ if $\pi_S(\mC)$, the projection map onto the coordinates indexed by $S$, has dimension $k$. Equivalently, given any generator matrix $G$ of $\mC$, $S \subseteq \{1, \dots, n \}$ is an information set for $\mC$ if the columns of $G$ indexed by $S$ form a matrix of rank $k$.
\end{definition}

\begin{definition}
    Let $\mathcal{C} \subseteq \F_q^n$ be a $k$-dimensional code and let 
    $G$ be a generator matrix of $\mC$.
    For $0 \leq s \leq n$, we denote by $g_j$ the $j$-th column of $G$. Define
    \[
    \alpha(G,s) = \left|\left\{ S \subseteq \{1, \ldots, n\} \, : \, |S| = s, \, \langle g_j \, : \, j \in S \rangle = \F_q^k\right\}\right|,
    \]
    which counts the number of information sets of cardinality $s$ of $\mC$.
\end{definition}

Following the same reasoning as Proposition~\ref{independence_from_G_proposition},
it can be checked that 
$\alpha(G,s)$ only depends
on the code $\mathcal{C}$ that $G$ generates. We will therefore use the symbols $\alpha(\mathcal{C},s)$ and $\alpha(G,s)$ interchangeably.

We start by establishing
an extension of \cite[Lemma~1]{a_combinatorial_perspective}, expressing $\mathbb{E}[\mC]$ in terms of the values $\alpha(\mathcal{C},s)$ we just introduced.
The proof is similar to that of \cite[Lemma~1]{a_combinatorial_perspective} and is therefore omitted.

\begin{proposition} \label{E_k[C]_general_formula}
For any $k$-dimensional code $\mC \subseteq \F_q^n$
we have 
    \[
    \mathbb{E}[\mathcal{C}] = nH_n -  \sum_{s=k}^{n-1} \frac{\alpha(\mathcal{C},s)}{\binom{n-1}{s}}.
    \]
\end{proposition}

\begin{remark}
    Note that the sum starts from $k$ because we need at least $k$ vectors for successfully recovering all the information strands. Equivalently, 
    \[
    \alpha(\mC,s) = 0 \, \text{ for } \, 0 \leq s \leq k-1.
    \]
\end{remark}

We illustrate how Proposition~\ref{E_k[C]_general_formula} can be used to easily compute the expectation for MDS codes.

\begin{example}
    Consider an $[n,k]_q$ MDS code $\mC$ and let $G$ be a generator matrix of $\mC$. Since $G$ is an MDS matrix, every $k$ columns of $G$ are linearly independent. Thus, we have that
    \[
    \alpha(G,s) = \begin{cases}
    0 & \text{ if } 0 \le s \le k-1, \\
    \binom{n}{s} & \text{ if } k \leq s \leq n.
    \end{cases}
    \]
    Hence, by substituting these values into the formula of Proposition~\ref{E_k[C]_general_formula}, we obtain
    \[
    \mathbb{E}[\mC] = nH_n - \sum_{s=k}^{n-1} \frac{\binom{n}{s}}{\binom{n-1}{s}},
    \]
    which simplifies to $n(H_n - H_{n-k})$ after straightforward computations.
\end{example}

We obtain a duality result by relating 
the value of $\alpha(\mathcal{C},s)$ to the structure of the dual code $\mC^\perp$. To do so,
it is convenient to introduce some auxiliary quantities. We denote the Hamming support of a vector $x \in \F_q^n$ as $\sigma(x)=\{i \, :\, x_i \neq 0\}$.
For a code $\mC \subseteq \F_q^n$ and a subset $S \subseteq \{1, \ldots, n\}$, we let $\mC(S)=\{x \in \mC \mid \sigma(x) \subseteq S\}$. The complement of a set $S$ is denoted by $S^{\mathsf{c}}=\{1, \ldots, n\} \setminus S$.

\begin{notation} \label{beta_definition}
For $1 \le \ell \le k$ and $0 \le s \le n$, let
\begin{equation*} 
\beta_{\ell}(\mathcal{C},s) = |\{ S \subseteq \{1, \ldots, n\} \, : \, |S| = s, \, \dim(\mathcal{C}(S^{\mathsf{c}}))=\ell \}|.
\end{equation*}
\end{notation}

The main tool of this section is the following result.

\begin{lemma}
    Let $\mathcal{C}$ be
    an $[n,k]_q$ code. We have
    \begin{equation} \label{beta_relation}
    \beta_{\ell}(\mathcal{C},s) = \beta_{\ell+s-k}(\mathcal{C}^\perp,n-s).
    \end{equation}
    In particular,
    \vspace{-0.1 cm}
    \begin{equation} \label{alpha_beta_formula}
    \alpha(\mathcal{C}, s) = \beta_{s-k}(\mathcal{C}^{\perp},n-s).
    \end{equation}
\end{lemma}
\begin{proof}
    We consider the projection map $\pi_S : \mathcal{C} \rightarrow \mathbb{F}_q^s$ onto the coordinates indexed by $S$. Using the rank-nullity theorem we obtain
    \begin{equation} \label{rank_nullity}
    \dim(\pi_S(\mathcal{C})) + \dim(\ker(\pi_S)) = k,
    \end{equation}
    which we can rewrite as
    \begin{equation} \label{rank_nullity_2}
        \dim(\pi_S(\mathcal{C})) + \dim(\mathcal{C}(S^{\mathsf{c}})) = k.
    \end{equation} 
    Moreover, by~\cite[Theorem~24]{ravagnani} we have 
    \[
    |\mathcal{C}(S)| = \frac{|\mathcal{C}|}{q^{n-s}} |\mathcal{C}^{\perp}(S^{\mathsf{c}})|,
    \]
    i.e.,
    \begin{equation} \label{rank_nullity_dual}
    \dim(\mathcal{C}(S)) = k-n+s+\dim(\mathcal{C}^{\perp}(S^{\mathsf{c}})).
    \end{equation}
    Therefore, from~\eqref{rank_nullity} we know that $\dim(\pi_S(\mathcal{C})) = t$ if and only if $\dim(\mathcal{C}(S^{\mathsf{c}})) = k-t$, and~\eqref{rank_nullity_dual} tells us that
    the latter equality is equivalent to $\dim(\mathcal{C}^{\perp}(S)) = s-t$. This shows that there is a bijection
    \vspace{-0.1 cm}
    \begin{multline*}
    \{ S \subseteq  \{1, \ldots, n\} \, : \,  |S| = s, \, \dim(\mathcal{C}(S^{\mathsf{c}})) = k-t \} \\  \quad \to \{ S \subseteq \{1, \ldots, n\} \, : \, |S| = n-s, \, \dim(\mathcal{C}^{\perp}(S^{\mathsf{c}})) = s-t  \},
    \end{multline*}
    from which we obtain the first part of the lemma. For the second part, it suffices to use the fact that $\alpha(\mathcal{C},s) = \beta_0(\mathcal{C},s)$,
    which easily follows from the definitions.
    Combining this equality with \eqref{beta_relation} we obtain the second part of the lemma.
\end{proof}

\section{Performance of the Hamming Code} \label{Hamming_section}

We wish to apply the result of Section~\ref{duality_section} to compute the value of $\E[\mC]$, where $\mC$ is a Hamming code, in order to obtain a solution to Problem~\ref{first_problem} also for this family of codes.

\begin{theorem} \label{thm_E_k_Hamming}
    Let $\mathcal{C} \subseteq \mathbb{F}_q^n$ be the $q$-ary Hamming code of redundancy $r$, where $n=(q^r-1)/(q-1)$. We have 
    \[
    \mathbb{E}[\mathcal{C}] = nH_n -  \sum_{\ell=1}^{r} \frac{1}{\binom{n-1}{n-\ell}}\frac{\prod_{i = 0}^{\ell-1} \frac{q^r-q^i}{q-1}}{\ell !}.
    \]
\end{theorem}
\begin{proof}
    Combining the formula in Proposition~\ref{E_k[C]_general_formula} with~\eqref{alpha_beta_formula}, where the dual code $\mathcal{C}^{\perp}$ is the $[n,r]_q$ simplex code, gives
    \[
    \mathbb{E}[\mathcal{C}] = n H_n - \sum_{s=n-r}^{n-1} \frac{\beta_{s-n+r}(\mathcal{C}^{\perp},n-s)}{\binom{n-1}{s}}.
    \]
    This can be rewritten as
    \[
    \mathbb{E}[\mathcal{C}] = n H_n - \sum_{\ell=1}^{r} \frac{\beta_{r-\ell}(\mathcal{C}^{\perp},\ell)}{\binom{n-1}{n-\ell}}.
    \]
Applying~\eqref{rank_nullity_2} to Notation~\ref{beta_definition} we obtain
    \begin{multline*}
    \beta_{r-\ell}(\mathcal{C}^\perp,\ell) = \\ |\{ S \subseteq \{1, \ldots, n\} : |S| = \ell, \dim(\pi_S(\mathcal{C}^\perp))=\ell \}|.
    \end{multline*}
    It remains to count the number of subsets of cardinality $\ell$ whose corresponding columns are linearly independent. To do this, we use again the fact that the columns of any generator matrix of the simplex code are all the non-zero vectors of $\mathbb{F}_q^r$ up to non-zero scalar multiples. Hence we have
    \[
    \beta_{r-\ell}(\mathcal{C}^\perp,\ell) = \frac{\prod_{i=0}^{\ell-1} \Big(\frac{q^r-1}{q-1} - \frac{q^i-1}{q-1}\Big)}{\ell!} = \frac{\prod_{i=0}^{\ell-1} \frac{q^r-q^i}{q-1}}{\ell!},
    \]
    from which the statement follows.
\end{proof}

\section{Asymptotic Estimates and Comparisons} \label{estimates_section}
In Theorems~\ref{thm_E_k_simplex} and~\ref{thm_E_k_Hamming} we provided closed formulas for $\mathbb{E}[\mC]$, where $\mC$ is a simplex or a Hamming code. 
We proceed by investigating
the difference between these values and the bound
of Theorem~\ref{thm_MDS_bound}, as the field size $q$ approaches infinity and the dimension $k$ is fixed.

\begin{notation} 
From now on we let
\begin{equation} \label{MDS_bound_formula}
    \widehat{\E}[n,k]_q \triangleq n(H_n - H_{n-k})
\end{equation}
be the lower bound stated in Theorem~\ref{thm_MDS_bound}.
We will use the Buchmann-Landau notation to describe the asymptotic growth of functions defined on an infinite set of natural numbers; see e.g.~\cite{asymp}.
\end{notation}

We start with simplex codes.

\begin{proposition} \label{prop3}
    Let $\mathcal{C} \subseteq \mathbb{F}_q^n$ be the $q$-ary simplex code of fixed dimension $k \ge 3$, where $n = (q^k-1)/(q-1)$.
    The following holds as $q \rightarrow \infty$: 
    \[
    \mathbb{E}[\mathcal{C}]-\widehat{\E}[n,k]_q = \frac{1}{q-1} + O \Big( \frac{1}{q^2} \Big)
    \]
    and
    \[
    \frac{\mathbb{E}[\mC]}{\widehat{\E}[n,k]_q} = \frac{k + \sum_{i=1}^k \frac{1}{q^i - 1} - \sum_{i=1}^k \frac{1}{q^k - q^{i-1}}}{k + \sum_{i=0}^{k-1} \frac{i}{\frac{q^k-1}{q-1}-i}}.
    \]
    In particular,
    \[
    \lim_{q \rightarrow \infty} \Big( \mathbb{E}[\mathcal{C}]-\widehat{\E}[n,k]_q \Big) = 0, \quad  
    \lim_{q \rightarrow \infty} \frac{\mathbb{E}[\mC]}{\widehat{\E}[n,k]_q} = 1.
    \]
\end{proposition}

\begin{proof}
For the parameters of simplex codes,~\eqref{MDS_bound_formula} can be rewritten as 
    \[
        \widehat{\E}[n,k]_q = \sum_{i=0}^{k-1} \frac{\frac{q^k-1}{q-1}}{\frac{q^k-1}{q-1}-i} = k + \sum_{i=0}^{k-1} \frac{i}{\frac{q^k-1}{q-1}-i}.
    \]
  Note that 
    \[
    \sum_{i=0}^{k-1} \frac{i}{\frac{q^k-1}{q-1}} \le \sum_{i=0}^{k-1} \frac{i}{\frac{q^k-1}{q-1}-i} \le \sum_{i=0}^{k-1} \frac{i}{q^{k-2}},
    \]
    i.e.,
    \[
    \frac{\binom{k}{2}}{\frac{q^k-1}{q-1}} \le \sum_{i=0}^{k-1} \frac{i}{\frac{q^k-1}{q-1}-i} \le \frac{\binom{k}{2}}{q^{k-2}}.
    \]
    Taking the limits as $q$ tends to infinity yields
    \[
    \lim_{q \rightarrow \infty} \sum_{i=0}^{k-1} \frac{i}{\frac{q^k-1}{q-1}-i} = 0.
    \]
    As for the expectation for the simplex code, by Theorem~\ref{thm_E_k_simplex} we have that
    \[
    \begin{split}
        \mathbb{E}[\mC]
        = k + \sum_{i=1}^k \frac{1}{q^i - 1} - \sum_{i=1}^k \frac{1}{q^k - q^{i-1}}.
    \end{split}
    \]
    Arguing in a similar way
    to what was done previously, it is not difficult to show that 
    \[
    \lim_{q \rightarrow \infty} \sum_{i=1}^{k} \frac{1}{q^k-q^{i-1}} = 0.
    \]
    Combining all of the above we obtain 
    \begin{equation} \label{difference_simplex_bound}
    \mathbb{E}[\mC] - \widehat{\E}[n,k]_q = \frac{1}{q-1} + O\Big( \frac{1}{q^2} \Big)
    \end{equation}
    and
    \begin{equation} \label{ratio_simplex_bound}
    \frac{\mathbb{E}[\mC]}{\widehat{\E}[n,k]_q} = \frac{k + \sum_{i=1}^k \frac{1}{q^i - 1} - \sum_{i=1}^k \frac{1}{q^k - q^{i-1}}}{k + \sum_{i=0}^{k-1} \frac{i}{\frac{q^k-1}{q-1}-i}}.
    \end{equation}
    Taking the limit as $q$ tends to infinity of both~\eqref{difference_simplex_bound} and~\eqref{ratio_simplex_bound} follows the second part of the statement.
\end{proof}

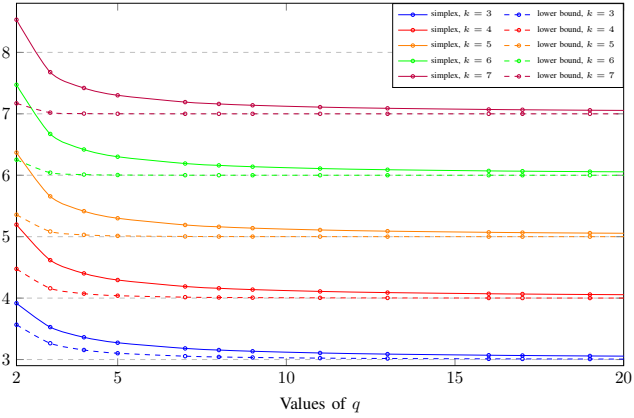
\begin{figure}[ht!]
\centering
\begin{tikzpicture}[scale=0.65]
\begin{axis}[legend columns = 2, legend style={at={(1,1)}, legend style={cells={align=left}}, anchor = north east, /tikz/column 2/.style={
                column sep=5pt}},
                legend cell align={left},
                width=14cm,height=9cm,
    xlabel={Values of $q$},
    xmin=2, xmax=20,
    ymin=2.9, ymax=8.8,
    xtick={2,5,10,15,20},
    ytick={3,4,5,6,7,8},
    ymajorgrids=true,
    grid style=dashed,
    every axis plot/.append style={},  yticklabel style={/pgf/number format/fixed}
]
\addplot+[color=blue,mark=o,mark size=1pt,smooth]
coordinates {
(2, 3.91666666666666666666666666667)
(3, 3.52777777777777777777777777778)
(4, 3.36250000000000000000000000000)
(5, 3.27333333333333333333333333333)
(7, 3.18112244897959183673469387755)
(8, 3.15451388888888888888888888889)
(9, 3.13456790123456790123456790123)
(11, 3.10674931129476584022038567493)
(13, 3.08833474218089602704987320372)
(16, 3.07008272058823529411764705882)
(17, 3.06555171088043060361399461745)
(19, 3.05803324099722991689750692521)
(23, 3.04718021424070573408947700063) 
    };
\addplot+[color=blue,mark=o,mark size=1pt,smooth,dashed,mark options={solid}]
coordinates {
(2, 3.56666666666666666666666666667)
(3, 3.26515151515151515151515151515)
(4, 3.15526315789473684210526315789)
(5, 3.10229885057471264367816091954)
(7, 3.05422077922077922077922077922)
(8, 3.04205790297339593114241001565)
(9, 3.03358302122347066167290886392)
(11, 3.02284293314827665972704140643)
(13, 3.01654422925141157185356080384)
(16, 3.01105654438897330149772085956)
(17, 3.00982535090538947819564984464)
(19, 3.00790862380224968754339675045)
(23, 3.00544135826823430390068123833) 
};
\addplot+[color=red,mark=o,mark size=1pt,smooth]
coordinates {
(2, 5.19642857142857142857142857143)
(3, 4.61823361823361823361823361823)
(4, 4.40252976190476190476190476190)
(5, 4.29445161290322580645161290323)
(7, 4.18909518694695923482174824817)
(8, 4.15991545376712328767123287671)
(9, 4.13839521246928654336061743469)
(11, 4.10887285832914367059647616411)
(13, 4.08963290726798341503938555059)
(16, 4.07078557835393772893772893773)
(17, 4.06613909384860083365875683141)
(19, 4.05845567962701265345185110354)
(23, 4.04741979126832116814357633839) 
    };
\addplot+[color=red,mark=o,mark size=1pt,smooth,dashed,mark options={solid}]
coordinates {
(2, 4.47527472527472527472527472528)
(3, 4.15935368566947514315935368567)
(4, 4.07258651330058911604606579629)
(5, 4.03904646902749369542538422994)
(7, 4.01508806635527339671152997362)
(8, 4.01029749967479954450135082334)
(9, 4.00733795962029735181554856973)
(11, 4.00410490414575310027563247320)
(13, 4.00252348265365651363829232803)
(16, 4.00137404584155640809357241924)
(17, 4.00114993933222318400656221811)
(19, 4.00082899646252873117451056199)
(23, 4.00047178465813502794467160377) 
};
\addplot+[color=orange,mark=o,mark size=1pt,smooth]
coordinates {
(2, 6.36964285714285714285714285714)
(3, 5.65671889838556505223171889839)
(4, 5.41547837885154061624649859944)
(5, 5.29995732009925558312655086849)
(7, 5.19059129237086886311989887255)
(8, 5.16080432459057780119423955040)
(9, 5.13895597062815772039901318231)
(11, 5.10912800443311343858396973566)
(13, 5.08976508674139220571214054043)
(16, 5.07084381229732228400938346006)
(17, 5.06618491466327120328726324066)
(19, 5.05848518279420970937401926180)
(23, 5.04743362576309580692252131917) 
    };
\addplot+[color=orange,mark=o,mark size=1pt,smooth,dashed,mark options={solid}]
coordinates {
(2, 5.35758985586571793468345192483)
(3, 5.08475181902400247365352051294)
(4, 5.02958605733208832746244129767)
(5, 5.01285349165615541002877612003)
(7, 5.00357398187171578424435757182)
(8, 5.00213766576675991606617746648)
(9, 5.00135538088691932719748808742)
(11, 5.00062104086688409754575227086)
(13, 5.00032322709969769775627720408)
(16, 5.00014305742327896948684184468)
(17, 5.00011269129348497739611149728)
(19, 5.00007269660798011367841976625)
(23, 5.00003418125636665835693799134) 
};
\addplot+[color=green,mark=o,mark size=1pt,smooth, solid]
coordinates {
(2, 7.47237903225806451612903225806)
(3, 6.67230214603951977689351426725)
(4, 6.41944867091814731757805761601)
(5, 6.30131454348468432975475228996)
(7, 6.19085602391023597021752318523)
(8, 6.16094213713930111292787537180)
(9, 6.13903333075657495018902054375)
(11, 6.10915684430782447193577230386)
(13, 6.08977774051442112524933841365)
(16, 6.07084834598930810165680139277)
(17, 6.06618827287286240840778468360)
(19, 6.05848711819786519549891519276)
(23, 6.04743437587550280150927742526) 
    };
\addplot+[color=green,mark=o,mark size=1pt,smooth,dashed,mark options={solid}]
coordinates {
(2, 6.25291942422518181882805784946)
(3, 6.04162861983535380181137639481)
(4, 6.01101861845045213649267716883)
(5, 6.00384385449695657309663557180)
(7, 6.00076513696261249021753080020)
(8, 6.00040058396284291988362643836)
(9, 6.00022581405976477200007298940)
(11, 6.00008467288981466139563397161)
(13, 6.00003729206565759806020039811)
(16, 6.00001341108983882133673914333)
(17, 6.00000994303085234770593976098)
(19, 6.00000573908658483540071380411)
(23, 6.00000222919039619210419776414) 
};
\addplot+[color=purple,mark=o,mark size=1pt,smooth, solid]
coordinates {
(2, 8.53168362775217613927291346646)
(3, 7.67841231283992226753169514112)
(4, 7.42062439411794917606113052033)
(5, 7.30163719143877980770302442270)
(7, 7.19090112835757459465631068832)
(8, 7.16096270158073192316684346132)
(9, 7.13904359893081141727208239503)
(11, 7.10915997927302427930866395958)
(13, 7.08977890512119623136268102007)
(16, 7.07084868522441502285774834990)
(17, 7.06618850940678614862152300253)
(19, 7.05848724019833314582356974049)
(23, 7.04743441495050363164319316437) 
    };
\addplot+[color=purple,mark=o,mark size=1pt,smooth,dashed,mark options={solid}]
coordinates {
(2, 7.17122076494373792463354292582)
(3, 7.01928968708832979895330497780)
(4, 7.00384850364698545108449426723)
(5, 7.00107545237909111668013345364)
(7, 7.00015300249177291424503022583)
(8, 7.00007009610955639156967602290)
(9, 7.00003512488849263585446723605)
(11, 7.00001077634499904669826934458)
(13, 7.00000401603455982636452541969)
(16, 7.00000117346673253479053470762)
(17, 7.00000081883595113020953641168)
(19, 7.00000042287949927845384495464)
(23, 7.00000013568977918633629477423) 
};
\legend{\tiny{simplex, $k=3$}, \tiny{lower bound, $k=3$},\tiny{simplex, $k=4$}, \tiny{lower bound, $k=4$},\tiny{simplex, $k=5$}, \tiny{lower bound, $k=5$},\tiny{simplex, $k=6$}, \tiny{lower bound, $k=6$},\tiny{simplex, $k=7$}, \tiny{lower bound, $k=7$}
}
\end{axis}
\end{tikzpicture}
\caption{\label{fig:averages} Expected number of reads $\mathbb{E}[\mC]$ for simplex codes from Theorem~\ref{thm_E_k_simplex} for various dimensions compared to the lower bound $\widehat{\E}[n,k]_q$.}
\end{figure}

For Hamming codes, we fix the value of the redundancy and compute an asymptotic estimate that allows us to understand how rapidly the expectation grows.

\begin{proposition}
    Let $\mC \subseteq \mathbb{F}_q^n$ be the $q$-ary Hamming code of fixed redundancy $r$, where $n = (q^r-1)/(q-1)$. The following holds as $q \rightarrow \infty$: 
    \[
    \mathbb{E}[\mC] - \widehat{\E}[n,n-r]_q \le \Big(H_r - \frac{r-1}{r}\Big)q^{r-2} + O(q^{r-3}). 
    \]
    Furthermore,
    \[
    \lim_{q\rightarrow\infty} \frac{\mathbb{E}[\mC]}{\widehat{\E}[n,n-r]_q} = 1.
    \]
\end{proposition}
\begin{proof}
    The difference between the formula obtained in Theorem~\ref{thm_E_k_Hamming} and~\eqref{MDS_bound_formula} gives
    \[
    \mathbb{E}[\mC] - \widehat{\E}[n,n-r]_q = nH_r - \sum_{\ell=1}^{r} \frac{1}{\binom{n-1}{n-\ell}}\frac{\prod_{i = 0}^{\ell-1} \frac{q^r-q^i}{q-1}}{\ell !}.
    \]
    After lenghty computations we obtain
    \[
    \begin{split}
        \mathbb{E}&[\mC] - \widehat{\E}[n,n-r]_q \\ &\le nH_r - q^{r-1}H_r - \Big( 1-\frac{1}{r} \Big)q^{r-2} + O(q^{r-3}),
    \end{split}
    \]
    which simplifies to
    \[
    \begin{split}
        \mathbb{E}[\mC] - \widehat{\E}[n,n-r]_q  
        \le \Big(H_r - \frac{r-1}{r}\Big)q^{r-2} + O(q^{r-3}).
    \end{split}
    \]
    The second part of the statement follows from the fact that the leading term of both $\mathbb{E}[\mC]$ and $\widehat{\E}[n,n-r]_q$, as $q$ goes to infinity, is $\smash{q^{r-1}\log\Big(\frac{q^r-1}{q-1}\Big)}$.
\end{proof}

It is also interesting to analyze the asymptotic behavior when the field size $q$ is fixed
and we let~$k$ go to infinity. The next two propositions do this for simplex and Hamming codes, respectively. Their proofs are omitted and will appear in the extended version of this work.

\begin{proposition}
    Let $\mC \subseteq \mathbb{F}_q^n$ be the $q$-ary 
    simplex code of dimension $k$, where $q$ is fixed and $n=(q^k-1)/(q-1)$. We have
    \[
    \lim_{k\rightarrow\infty} \Big(\mathbb{E}[\mC] - \widehat{\E}[n,k]_q\Big) = \sum_{i=1}^{\infty} \frac{1}{q^i -1} 
    \]
    and
    \[
    \lim_{k\rightarrow\infty} \frac{\mathbb{E}[\mC]}{ \widehat{\E}[n,k]_q} = 1.
    \]
\end{proposition}

The last result we present holds in the binary case.

\begin{proposition}
    Let $\mC \subseteq \mathbb{F}_2^n$ be the binary Hamming code of redundancy $r$. The following holds as $r \rightarrow \infty$:
    \[
    \mathbb{E}[\mC] - \widehat{\E}[n,n-r]_2 \le (H_{2^r-1} - H_r - 1)2^r + O(2^{r-1}). 
    \]
    Furthermore,
    \[
    \lim_{r\rightarrow\infty} \frac{\mathbb{E}[\mC]}{\widehat{\E}[n,n-r]_2} \leq H_{2^r-1}-H_r.
    \]
\end{proposition}

\section{Conclusions and Future Work}
We studied the coverage depth problem, which aims to reduce sequencing costs in DNA storage systems,
while ensuring efficiency and high-accuracy retrieval. We focused on codes defined over small fields, solving the problem
for various code families, and comparing the values we obtained with theoretical bounds. The computations rely on a duality result linking the performance of a code to that of the dual code.

Future work will focus on determining if simplex codes offer the best performance for their parameters, 
and on solving the coverage depth problem for other code families defined over small fields, such as Reed-Muller and Golay codes.

\end{document}